\documentclass[letterpaper,11pt]{article}

\usepackage{times}
\usepackage{amsmath}
\usepackage{amsthm}
\usepackage{amssymb}
\usepackage{amsfonts}
\usepackage{fullpage}

\newcommand{\sketch}[1]{}

\newtheorem{theorem}{Theorem}[section]
\newtheorem{lemma}[theorem]{Lemma}

\newtheorem{claim}[theorem]{Claim}

\newtheorem{fact}[theorem]{Fact}
\newcommand{\st}{\mbox{\rm s.t. }}

\def \M   {{\mathcal M}}
\def \I   {{\mathcal I}}
\def \A   {{\mathcal A}}
\def \mZ  {{\mathbb Z}}
\def \N   {{\mathcal N}}

\date{}%
\newenvironment{proofof}[1]{\noindent{\bf Proof of #1:}}%
            {\hspace*{\fill}$\Box$\par\vspace{4mm}}

\newcommand{\removefromfile}[1]{}

\begin{document}
\title{{\bf Approximation Algorithms for Online Weighted Rank Function Maximization under Matroid Constraints} }
\date{}

\author{
Niv Buchbinder \thanks{Computer Science Dept., Open University of Israel. E-mail: niv.buchbinder@gmail.com} \and
  \and Joseph
(Seffi) Naor\thanks{Computer Science Dept., Technion, Haifa,
Israel. E-mail: naor@cs.technion.ac.il} \and R. Ravi\thanks{Tepper School of Business,
 Carnegie Mellon University, Pittsburgh, PA. E-mail: ravi@cmu.edu. Supported in part by NSF award CCF-1143998}
\and Mohit Singh \thanks{Microsoft Research, Redmond and School of Computer Science, McGill University,
 Montreal, Quebec, Canada. E-mail: mohitsinghr@gmail.com}}
 \maketitle

\begin{abstract}
Consider the following online version of the submodular maximization problem under a matroid constraint: We are given a set of elements over which a matroid is defined. The goal is to incrementally choose a subset that remains independent in the matroid over time. At each time, a new weighted rank function of a different matroid (one per time) over the same elements is presented; the algorithm can add a few elements to the incrementally constructed set, and reaps a reward equal to the value of the new weighted rank function on the current set. The goal of the algorithm as it builds this independent set online is to maximize the sum of these (weighted rank) rewards. As in regular online analysis, we compare the rewards of our online algorithm to that of an offline optimum, namely a single independent set of the matroid that maximizes the sum of the weighted rank rewards that arrive over time. This problem is a natural extension of two well-studied streams of earlier work: the first is on online set cover algorithms (in particular for the max coverage version) while the second is on approximately maximizing submodular functions under a matroid constraint.

In this paper, we present the first randomized online algorithms for this problem with poly-logarithmic competitive ratio. To do this, we employ the LP formulation of a scaled reward version of the problem. Then we extend a weighted-majority type update rule along with uncrossing properties of tight sets in the matroid polytope to find an approximately optimal fractional LP solution. We use the fractional solution values as probabilities for a online randomized rounding algorithm. To show that our rounding produces a sufficiently large reward independent set, we prove and use new covering properties for randomly rounded fractional solutions in the matroid polytope that may be of independent interest.

\end{abstract}

\section{Introduction}

Making decisions in the face of uncertainty is the fundamental problem addressed by online computation \cite{BE98}. In many planning scenarios, a planner must decide on the evolution of features to a product without knowing the evolution of the demand for these features from future users. Moreover, any features initially included must be retained for backward compatibility, and hence leads to an online optimization problem: given a set of features, the planner must phase the addition of the features, so as to maximize the value perceived by a user at the time of arrival. Typically, users have diminishing returns for additional features, so it is natural to represent their utility as a submodular function of the features that are present (or added) when they arrive. Furthermore, the set of features that are thus monotonically added, are typically required to obey some design constraints. The simplest are of the form that partition the features into classes and there is a restriction on the number of features that can be deployed in each class. A slight extension specifies a hierarchy over these classes and there are individual bounds over the number of features that can be chosen from each class. We capture these, as well as other much more general restrictions on the set of deployed features, via the constraint that the chosen features form an independent set of a matroid. Thus, our problem is to monotonically construct an independent set of features (from a matroid over the features) online, so as to maximize the sum of submodular function values (users) arriving over time and evaluated on the set of features that have been constructed so far.

This class of online optimization problems
generalizes some early work of Awerbuch et al.~\cite{AAFL96}. They considered a set-cover instance, in which the restriction is to choose at most $k$ sets with the goal of maximizing the coverage of the elements as they arrive over time. This is precisely the online maximization version of the well-studied maximum coverage problem. Even this special case of our problem already abstracts problems in investment planning, strategic planning, and video-on-demand scheduling.

\subsection{Problem Setting, Main Result and Techniques}\label{sec:results}

In our setting\footnote{For preliminaries and basic definitions, please see Section \ref{sec:prelim}.}, we are given a universe of elements $E$, $|E|=m$, and a matroid $\M=(E,\I(\M))$ whose independent sets characterize the limitations on which sets of elements we can choose. At every time step $i$, $1\leq i \leq n$, a client arrives with a non-negative monotone submodular function $f_i:2^{E}\rightarrow \mZ_+$ representing her welfare function. The objective is to maintain a monotonically increasing set $F \in \I(\M)$ over time; that is, the set $F_{i-1}$ of elements (at time $i-1$) can only be augmented to $F_i$ after seeing $f_i$ at time step $i$.
The welfare of client $i$ is then $f_i(F_i)$, and our objective is to maximize $\sum_{i=1}^{n} f_i(F_i)$.
We compare our performance to the offline optimum $\max_{O\in \I(\M)} \sum_{i=1}^{n}f_i(O)$.

In this paper, we are concerned with the case when each of the submodular functions $f_i$ is a weighted rank function of some matroid $\N_i$, i.e., $f_i(S)= \max_{I\subseteq S, I\in \I(\N_i)}\sum_{e\in I} w_{i,e}$ where $w_i: E\rightarrow {\mathbb R}_+$ is an arbitrary weight function. This class of submodular functions is very broad and includes all the examples discussed above; Furthermore, we believe it captures the difficulty of general submodular functions even though we have not yet been able to extend our results to the general case. Nevertheless, there are submodular functions which are not weighted rank functions of a matroid, for example, multi-set coverage function~\cite{CCPV07}.

\begin{theorem}\label{theorem:main1}
There exists a randomized polynomial time algorithm which is $O\left(\log^2 n  \log m \log f_{ratio}\right)$-competitive,
for the online submodular function maximization problem under a matroid constraint over $m$ elements, when each $f_i$, $1\leq i \leq n$, is a weighted rank function of a matroid and  $f_{ratio} =2\frac{\max_{i,e} f_i(\{e\})}{\min_{i,e| f_i(\{e\})\neq 0}f_i(\{e\})}$. In other words, the algorithm maintains monotonically
increasing independent sets $F_i\in \I(\M)$ such that
$$E\left[\sum_{i=1}^n f_i(F_i)\right] \geq \Omega\left(\frac{1}{\log^2 n  \log m \log f_{ratio}}\right)\cdot  \max_{O\in \I(\M)} \sum_{i=1}^n f_i(O).$$
\end{theorem}

Our result should be contrasted with the lower bound proved in \cite{AAFL96}\footnote{The lower bounds in \cite{AAFL96} relate even to a special case of uniform matroid and very restricted sub-modular functions.}.

\begin{lemma}\cite{AAFL96}
Any randomized algorithm for the submodular maximization problem under a matroid constraint is $\Omega(\log n \log (m/r))$-competitive, where $r$ is the rank of the matroid. This lower bound holds even for uniform matroids and when all $f_i$ are unweighted rank functions.
\end{lemma}

We note that the $O(\log m)$ factor in our analysis can be improved slightly to an $O(\log (m/r))$ factor with a more careful analysis. A lower bound of $\Omega(\log f_{ratio})$ also follows even when the functions $f_i$ are linear (see, for example, \cite{BN09}).

\vspace{0.15cm}\noindent
{\bf Main Techniques.} To prove our results, we combine techniques from online computation and combinatorial optimization. The first step is to formulate an
integer linear programming formulation for the problem. Unfortunately, the natural linear program is not well-suited for the online
version of the problem. Thus, we formulate a different linear program in which we add an extra constraint that each element
$e$ \emph{contributes} roughly the same value to the objective of the optimal solution. While this may not be true in general,
we show that an approximate optimal solution satisfies this requirement.

We note that the online setting we study is quite different from
the online packing framework studied by \cite{BN09} and leads to new technical challenges. 
In particular, there are two obstacles in applying the primal-dual techniques in \cite{BN09} to our setting.
First, the linear formulation we obtain (that is natural for our problem) is not a strict packing LP and contains negative variables (See Section \ref{sec:frac}).
Second, the number of packing constraints is exponential, and hence the techniques of \cite{BN09}
would give a linear competitive factor rather than a polylogarithmic
one.
Nevertheless, we present in Section~\ref{sec:frac} an online algorithm which gives a fractional
solution to the linear program having a large objective value. One of
the crucial ingredients is the \emph{uncrossing} property of
tight sets for any feasible point in the matroid polytope.

To obtain an integral solution, we perform in Section~\ref{sec:randomized} a natural randomized rounding procedure to select fractionally chosen elements. But, we
have to be careful to maintain that the selected elements continue to form an independent set. The main challenge in the analysis
is to tie the performance of the randomized algorithm to the performance of the fractional algorithm. 
As a technical tool in our proof, we show in Lemma~\ref{lem:matroid-cover} that randomly rounding a fractional solution in the matroid polytope gives a set which can be covered by $O(\log n)$ independent sets with high probability. This lemma may be of independent interest and similar in flavor to the results of Karger~\cite{Karger98} who proved a similar result for packing bases in the randomly rounded solution.

\subsection{Related results}
Maximizing monotone submodular function under matroid constraints has been a well studied problem and even many special cases have been studied widely (see survey by Goundan and Schulz~\cite{GoundanS07}). Fisher, Nemhauser and Wolsey~\cite{FNW2} gave a $(1-\frac1e)$-approximation when the matroid is the uniform matroid and showed that the greedy algorithm gives a $\frac 12$-approximation. This was improved by Calinescu {\em at al}~\cite{CCPV07} and Vondr\'{a}k~\cite{V08} who gave a $(1-\frac1e)$-approximation for the general problem. They also introduced the \emph{multi-linear extension} of a submodular function and used pipage rounding introduced by Ageev and Sviridenko~\cite{AgeevS04}. The facility location problem was introduced by Cornuejols et al.~\cite{CFN77} and was the impetus behind studying the general submodular function maximization problem subject to matroid constraints. The submodular welfare problem can be cast as a submodular maximization problem subject to a matroid constraint and the reduction appears in Fisher {\em et al.}~\cite{FNW2} and the problem has been extensively studied ~\cite{NemhauserW78,LehmannLN01,MirrokniSV08,KhotLMM08}. The result of Vondr\'{a}k~\cite{V08} implies a $(1-\frac1e)$-approximation for the problem. Despite the restricted setting of our benefit functions, we note that recent work in welfare maximization in combinatorial auctions~\cite{DughmiRY11} has focused on precisely the case when the valuations are matroid rank sums (MRS) that we consider in our model.

A special case of our online problem was studied by Awerbuch et al. \cite{AAFL96}. They studied an online variant of the max-coverage problem, where given $n$ sets covering $m$ elements, the elements arrive one at a time, and the goal is to pick up to $k$ sets online to maximize coverage.
 They obtained a randomized algorithm that is $O(\log n \log (m/k))$-competitive for the problem and proved that this is optimal in their setting. Our results generalize both the requirement on the cardinality of the chosen sets to arbitrary matroid constraints, and the coverage functions of the arriving elements to monotone submodular functions that are weighted rank functions of matroids.

Another closely related problem with a different model of uncertainty was studied by Babaioff et al. \cite{BIK07}.
They studied a setting in which elements of a matroid arrive in an online fashion and the goal is to construct an independent set that is competitive with the maximum weight independent set. They considered the random permutation model which is a non-adversarial setting, and obtained an $O(\log k)$-competitive algorithm for general matroid, where $k$ is the rank of the matroid, and constant competitive ratio for several interesting matroids. Recently, Bateni et al. \cite{BHZ10} studied the same model where the objective function is a submodular function (rather than linear). 

Chawla et al.~\cite{CHMS10} study Bayesian optimal mechanism design to maximize expected revenue for a seller while allocating items to agents who draw their values for the items from a known distribution. Their development of agent-specific posted price mechanisms when the agents arrive in order, and the items allocated must obey matroid feasibility constraints, is similar to our setting. In particular, we use the ideas about certain ordering of matroid elements (Lemma 7 in their paper) in the proof that our randomized rounding algorithm give sufficient profit.

\section{Preliminaries}
\label{sec:prelim}

Given a set $E$, a function $f:2^{E}\rightarrow {\mathbb R}_+$ is called \emph{submodular} if for all sets $A,B\subseteq E$, $f(A)+f(B)\geq f(A\cap B)+f(A\cup B)$.
Given set $E$ and a collection $\I\subseteq 2^{E}$, $\M=(E,\I(\M))$ is a \emph{matroid} if (i) for all $A\in \I$ and $B\subset A$ implies that  $B\in \I$ and
(ii) for all $A,B\in \I$ and $|A|> |B|$ then there exists $a\in A\setminus B$ such that $B\cup \{a\}\in \I$.  Sets in $\I$ are called {\em independent sets}
of the matroid $\M$. The rank function $r:2^{E}\rightarrow R^+$ of matroid $\M$ is defined as $r(S)=\max_{T\in \I: T\subseteq S} |T|$. A standard property of matroids is the fact
that the rank function of any matroid is submodular.

We also work with weighted rank functions of a matroid, defined as $f(S)=\max_{I\subseteq S, I\in \I(\M)}\sum_{e\in I} w_{e}$ for some weight function $w: 2^{E}\rightarrow {\mathbb R}_+$. Given any matroid $\M$, we define the matroid polytope to be the convex hull of independent sets $P(\M)=conv\{\bf{1}_{I}: I\in \I\}\subseteq {\mathbb R}^{|E|}$.
Edmonds~\cite{E70} showed that  $P(\M)=\{x\geq 0: x(S)\leq r(S)\;\; \forall\; S\subseteq E\}$. We also use the following fact about fractional
points in the matroid polytope (The proof follows from standard uncrossing arguments. See Schrijver~\cite{Schrijver}, Chapter 40).

\begin{fact}\label{fact:uncross}
Given a matroid $\M=(E,\I(\M))$ with rank function $r$ and feasible point $x\in P(\M)$, let $\tau=\{S\subseteq E: x(S)=r(S)\}$. Then,
 $\tau$ is closed under intersection and union and there is a single maximal set in $\tau$.
\end{fact}

\section{Linear Program and the Fractional Algorithm}\label{sec:frac}

We now give a linear program for the online submodular function maximization problem and show how to construct a feasible fractional solution online which is $O(\log m \log n\log f_{ratio})$-competitive.
Before we give the main theorem, we first formulate a natural LP.  Let $O\subseteq E $ denote the optimal solution with the objective $\sum_{i=1}^n f_i(O)$. Since each $f_i$ is the weighted rank function of matroid $\N_i$, we have that $f_i(O)=w_i(O_i)=\sum_{e\in O_i} w_{i,e}$ where $O\supseteq O_i \in \I(\N_i)$. For the sake of simplicity, we assume that $w_{i,e}=1$. In Section~\ref{sec:weighted}, we show that this assumption can be removed with a loss of $O(\log f_{ratio})$ factor in the competitive ratio. Observe that in this case, $f_i(S)=r_i(S)$, where $r_i$ is the rank function of matroid $\N_i$ for any set $S\subseteq E$.

We next formulate a linear program where $x_e$ is the indicator variable for whether $e\in O$ and $z_{i,e}$ is the indicator variable for whether $e\in O_i$. Since $O\in \I(M)$ and $O_i \in \I(\N_i)$, we have that $x\in P(\M)$ and $z_i\in P(\N_i)$ as represented by constraints (\ref{cons:matm}) and constraints (\ref{cons:matn}), respectively in Figure \ref{fig:lp-submodular2}.

\begin{figure}[t]
\begin{center}
\begin{eqnarray}
LP_1: \qquad \max & \sum_{i=1}^n\sum_{e \in E}z_{i,e} \nonumber\\
\st \nonumber \\
\forall S\subseteq E & \sum_{e \in S}x_{e}\leq r(S) \label{cons:matm}\\
\forall 1 \leq i \leq n, S\subseteq E & \sum_{e \in S} z_{i,e} \leq r_i(S) \label{cons:matn}\\
\forall 1 \leq i \leq n ,e \in E & z_{i,e} \leq x_e \\
\forall 1 \leq i \leq n ,e \in E & z_{i,e}, x_e \geq 0 \nonumber
\end{eqnarray}
\end{center}
\caption{LP for maximizing a sum  of (unweighted) rank functions subject to matroid constraint} \label{fig:lp-submodular2}
\end{figure}

We prove the following theorem.

\begin{theorem}\label{theorem:fractional}
There exists a polynomial time algorithm $\A$ that constructs a feasible fractional solution $(x,z)$ online to $LP_1$ which is $O(\log n\log m)$-competitive. That is, the algorithm $\A$ maintains monotonically increasing solution $(x,z)$ such that  $\sum_{i=1}^n\sum_{e\in E}z_{i,e} = \Omega(\frac{\sum_{i=1}^n f_i(O)}{\log n \log m})$ where $O$ is the optimal integral solution.
\end{theorem}

To prove Theorem~\ref{theorem:fractional}, instead of working with the natural linear program $LP_1$, we formulate a different linear program. The new linear program is indexed by an integer $\alpha$ and places the constraints that each $e\in O$ occurs in $[\frac{\alpha}{2},\alpha]$ different $O_i$'s as represented by constraints (\ref{added-cons11}) and (\ref{added-cons22}). The parameter $\alpha$ will be defined later.

\begin{figure}[t]
\begin{center}
\begin{eqnarray}
LP_2(\alpha): \qquad \max & \sum_{i=1}^n\sum_{e \in E}z_{i,e} \nonumber\\
\st \nonumber \\
\forall S\subseteq E & \sum_{e \in S}x_{e}\leq r(S) \label{cons-mat}\\
\forall 1 \leq i \leq n, S\subseteq E & \sum_{e \in S} z_{i,e} \leq r_i(S) \label{constri}\\
\forall 1 \leq i \leq n ,e \in E & z_{i,e} \leq x_e \label{constzexe}\\
\forall e \in E &  \sum_{i=1}^{n}z_{i,e} \leq \alpha x_e \label{added-cons11} \\
\forall e \in E & \sum_{i=1}^{n}z_{i,e} \geq \frac{\alpha x_e}{2} \label{added-cons22} \\
\forall 1 \leq i \leq n ,e \in E & z_{i,e}, x_e \geq 0  \nonumber
\end{eqnarray}
\end{center}
\caption{A restricted LP for the submodular function maximization subject to matroid constraint} \label{lp-rest2}
\end{figure}

The following lemma shows that if we pick $O(\log n)$ different values of $\alpha$ then the sum of the integer solutions to the linear programs $LP_2(\alpha)$ perform as well as the optimal solution.\footnote{We assume that the algorithm knows the value of $n$. In Section~\ref{sec:double} we show how to deal with an unknown $n$ losing an additional small factor.}
\begin{lemma}\label{lem:reduction}
Let $OPT$ denote the value of an optimal integral solution to linear program $LP_1$ and let $OPT_\alpha$ denote the optimal value of an optimal integral solution to the linear program $LP_2(\alpha)$ for each $\alpha\in \{1,2,4,\ldots, 2^{\lceil \log n\rceil}\}$. Then $OPT\leq \sum_{\alpha\in \{1,2,4,\ldots, 2^{\lceil \log n \rceil}\}} OPT_\alpha$.
\end{lemma}
\begin{proof}
Consider the optimal (integral) solution $(x^*,z^*)$ to $LP_1$. We decompose this solution in to $\lceil \log n \rceil$ integral solutions $(y(\alpha),z(\alpha))$ for each $\alpha\in \{1,2,4,\ldots, 2^{\lceil \log n \rceil}\}$ where $(y(\alpha),z(\alpha))$ is feasible to $LP_2(\alpha)$. Set $y(\alpha)_e=x^*_e$ for each $e\in E$ and each $\alpha$. Consider any $e\in E$. If $\alpha/2<\sum_{i=1}^n z^*_{i,e}\leq \alpha$ then set $z(\alpha)_{i,e}=1$ for each $i$ such that $z^*_{i,e}=1$. Set all other $z(\alpha)_{i,e}=0$. Clearly, $(y(\alpha),z(\alpha))$ is feasible and is a decomposition of the solution $(x^*,z^*)$ and hence sum of the objectives of $(y(\alpha),z(\alpha))$ equals the objective of $(x^*,z^*)$.
\end{proof}

Using the above lemma, a simple averaging argument shows that for some guess $\alpha$, the optimal integral solution to $LP_2(\alpha)$ is within a $\log n$ factor of the optimal integral solution to $LP_1$. Hence, we construct an algorithm which first guesses $\alpha$ and then constructs an approximate fractional solution to $LP_2(\alpha)$.

\subsection{Online Algorithm for a Fractional LP Solution}
Given a fractional solution $x$, we call a set $S\subseteq E$ tight (with respect to $x$) if $x(S)=r(S)$.

\vspace{0.15cm}
\noindent \framebox[1.05\width][l]{
\begin{minipage}{0.94\linewidth}
{\bf Guessing Algorithm:}
\begin{itemize}\setlength{\itemsep}{1pt}\setlength{\parskip}{0pt}
\item Guess the value $\alpha\in_R\{1, 2, 4 \ldots, n\}$.
\item Run AlgG with value $\alpha$.
\end{itemize}
\end{minipage}}
\vspace{0.15cm}

\vspace{0.15cm}
\noindent \framebox[1.05\width][l]{
\begin{minipage}{0.94\linewidth}
{\bf AlgG:} 
\begin{itemize}\setlength{\itemsep}{1pt}\setlength{\parskip}{0pt}
\item Initialize $x_e \gets 1/m^2$ (where $m=|E|$), set $z_{i,e}=0$ for each $i,e$.
\item When function $f_i$ arrives, order the elements arbitrarily.
\item For each element $e$ in order:
\item If $\forall S| e\in S, \ x(S) < r(S) \mbox{ and }  z_i(S) < r_i(S) - 1/2$:
\begin{eqnarray} x_e &\gets& \min \left\{x_e \cdot \exp \left(\frac{8 \log m}{\alpha}\right), \min_{S| e\in S}\{r(S)-x(S\setminus \{e\})\}\right\}\label{mult-label22}\\
z_{i,e} &\gets& x_e/2 \label{assignz}\end{eqnarray}
\end{itemize}
\end{minipage}}
\vspace{0.15cm}

Using an independence oracle for each of the matroids, each of the conditions can be checked in polynomial time by reducing it to submodular function minimization (See Schrijver~\cite{Schrijver}, Chapter 40) and therefore the running time of the algorithm is polynomial. Note that the fractional algorithm is carefully designed. For example, it is very reasonable to update greedily the value of $z_{i,e}$ even when the algorithm does not update the value $x_e$ (of course, ensuring that $z_i\in P(\N_i)$). While such an algorithm does give the required guarantee on the performance of the fractional solution, it is not clear how to round such a solution to an integral solution. In particular, our algorithm for finding a fractional solution  is tailored so as to allow us to use the values as rounding probabilities in a randomized algorithm.

Before we continue, we define some helpful notation regarding the online algorithm. Let $x_{i,e}(\alpha)$ be the value of the variable $x_e$ after the arrival of user $i$ for some guess $\alpha$. Let $\Delta x_{i,e}(\alpha)$ be the change in the value of $x_e$ when user $i$ arrives. Let $x_e(\alpha)$ be the value of $x_e$ at the end of the execution. Similarly, let $z_{i,e}(\alpha)$ be the value of $z_{i,e}$ at the end of the execution.
We start with the following lemma that follows from the update rule (\ref{mult-label22}). 

\begin{lemma}\label{obsxz}
For any element $e\in E$, and guess $\alpha$,
\begin{eqnarray}
\sum_{i=1}^{n} z_{i,e}(\alpha) \geq  \frac{\alpha}{48 \log m}\left(x_e(\alpha)-\frac{1}{m^2}\right)\label{z_x}\end{eqnarray}
where $x_e(\alpha)$ is the value at the end of the execution of AlgG.
\end{lemma}

\begin{proof}
Fix $e$, if $\alpha \leq 8\log m$, then the consider the last $z_{i,e}$ that we update (if there are none, the claim follows trivially). In this iteration $ z_{i,e} \gets x_e(\alpha)/2 \geq \frac{\alpha}{16\log m}\left(x_e(\alpha)-\frac{1}{m^2}\right)$.
Otherwise, if $\alpha \geq 8\log m$ then,

\begin{eqnarray}
\sum_{i=1}^{n} z_{i,e}(\alpha) & = & \sum_{i: \Delta x_{i,e}(\alpha)>0}x_{i,e}(\alpha)/2 \geq \sum_{i: \Delta x_{i,e}(\alpha)>0}x_{i-1,e}(\alpha)/2 \nonumber \\
& \geq & \frac{\alpha}{48\log m} \sum_{i: \Delta x_{i,e}(\alpha)>0} \Delta x_{i,e}(\alpha)  =  \frac{\alpha}{48\log m}\left(x_e(\alpha)-\frac{1}{m^2}\right) \label{ineqw1}
\end{eqnarray}

where Inequality (\ref{ineqw1}) follows since $e^x -1 \leq 3x$ for $0\leq x \leq 1$, and therefore, $\Delta x_{i,e} \leq x_{i,e}\cdot \left(\exp \left(\frac{8 \log m}{\alpha}\right)-1\right) \leq x_{i,e}\frac{24\log m}{\alpha}$ if $x_{i,e}$ is updated by the first condition in the equation (\ref{ineqw1}). If it is updated by the second condition, the inequality still holds since $\Delta x_{i,e}$ is even smaller.
\end{proof}

Next we prove that the solution produced by AlgG is almost feasible with respect to the optimal solution to $LP_2(\alpha)$.

\begin{lemma}[\textbf{Feasibility Lemma}]\label{lem-feasible}
Let $(x(\alpha),z(\alpha))$ be the fractional solution generated by AlgG at the end of the sequence. Then, it satisfies all constraints of $LP_2(\alpha)$ except constraints (\ref{added-cons22}).
\end{lemma}
\begin{proof}
We prove that the solution is feasible.
\vspace{-.4cm}
\paragraph{Matroid constraints (\ref{cons-mat}).}
Clearly, the algorithm never violates the matroid constraints by the second term in the equation (\ref{mult-label22}) in the algorithm.
\vspace{-.4cm}
\paragraph{Constraints (\ref{constri}) and constraints (\ref{constzexe}).}
$z_{i,e}\gets x_{i,e}(\alpha)/2 \leq x_e(\alpha)/2$, thus constraints (\ref{constzexe}) hold.
Finally, by the algorithm behavior we only update $z_{i,e}$ if for all $S | e\in S$, $z_i(S) < r_i(S) - 1/2$. Since by the above observations $z_{i,e} \leq x_e(\alpha)/2 \leq 1/2$, we never violate constraints (\ref{constri}) after the update.
\vspace{-.4cm}
\paragraph{Constraints (\ref{added-cons11}).}
This constraint follows since $$\sum_{i=1}^{n} z_{i,e} = \sum_{i:\Delta x_{i,e}>0}x_{i,e}(\alpha)/2\leq x_e(\alpha)|\{i:\Delta x_{i,e}>0\}|$$
However, after $\alpha$ augmentations, $x_e(\alpha) \geq \frac{1}{m^2}\exp\left(\frac{8\log m}{\alpha} \cdot \alpha\right) >1$. Thus, $x_e$ must be in a tight set and so by design we never update $x_e$ and any $z_{i,e}$.
\end{proof}

In order to evaluate the performance of the algorithm we first show that the \emph{size} of the solution returned by the algorithm is large as compared to the optimal integral solution. Later in Lemma~\ref{lem:guess}, we relate the objective value of the solution to its size. This lemma uses crucially the properties of the matroid. 
\begin{lemma}[\textbf{Large Fractional Size}]\label{lem:fractional-value}
Let $(x^*(\alpha),z^*(\alpha))$ be an optimal integral solution to $LP_2(\alpha)$. Let $(x(\alpha),z(\alpha))$ be the fractional solution generated by AlgG at the end of the sequence. Then we have,
$\sum_{e \in E} x_e(\alpha)\geq \frac{1}{16}\sum_{e \in E} x^*_e(\alpha)$.
\end{lemma}

\begin{proof}
For any element $e\in E$, let $P_e = \{i \ | \ z^*_{i,e}(\alpha)=1\}$. By constraint (\ref{added-cons22}) of the linear formulation $|P_e| \geq \alpha/2$. Observe that there can be at most $\frac{\alpha}{4}$ iterations where $\Delta x_{i,e}(\alpha)>0$ or equivalently $z_{i,e}>0$. Otherwise after $\frac{\alpha}{4}$ updates, the value of $x_e(\alpha)$ would be
 \begin{eqnarray}
\frac{1}{m^2}\exp\left(\frac{8\log m (\alpha/4)}{\alpha}\right) > 1
\end{eqnarray}
which is a contradiction.
By this observation we get that there are two possibilities for any element $e$ such that $x^*_e(\alpha)=1$. Since $\alpha$ is fixed we use $x$ instead of $x(\alpha)$.

\begin{enumerate}
\item There exists $S \subseteq E$, $e\in S$ such that $x(S) = r(S)$ (recall that we term such elements ``tight").
\item There exists at least $\frac{\alpha}{4}$ iterations $i\in P_e$ such that $z_{i,e}(S_i)\geq r_i(S_i)-\frac{1}{2}$ for some $S_i$ containing $e$.
\end{enumerate}

Let $n_\alpha$ be the number of elements such that $x^*_e(\alpha)=1$. Then, either the first condition holds for $n_\alpha/2$ elements, or the second condition holds for $n_\alpha/2$ elements. We prove that the Lemma holds in each of these options.

\paragraph{First condition holds for $n_\alpha/2$ elements:}
In this case, let $G$ be the set of all elements $e$ with $x^*_e(\alpha)=1$ for which there exists a set $S$ containing $e$ with $x(S)=r(S)$. The first condition implies that $|G|\geq \frac{n_\alpha}{2}$. Let $\tau=\{S: x(S)=r(S)\}$ be the set of all tight sets. Fact~\ref{fact:uncross} implies that there is a single maximal set in $\tau$, say, $\bar{S}$. Moreover, each element of $G$ is in one of the sets and $G\in \I(\M)$. Therefore
\begin{eqnarray*}\sum_{e\in E} x_e\geq \sum_{e\in \bar{S}}x_e=  r(\bar{S}) \geq |G|\geq \frac{n_\alpha}{2}
\end{eqnarray*}
Here the first inequality follows by summing only on elements that are in maximal tight sets. The second inequality follows since the maximal tight set contain $n_\alpha/2$ elements of the optimal solution that satisfies the matroid constraints.

\paragraph{Second condition holds for $n_\alpha/2$ elements:}
In this case, we have that there are at least $\alpha/4$ $i \in P_e$ there exists a set $S_i\subseteq E, e\in S_i$ such that $z_i(S_i) \geq r_i(S_i)-1/2$.

We next define for each iteration $i$ the following set:
$$G_i= \{e \ | \  e \mbox{ is {\bf not} tight }, z^*_{i,e}=1, z_{i,e}=0\}$$
Then we have,
$\sum_{i=1}^{n}|G_i| \geq \frac{n_\alpha}{2} \cdot \frac{\alpha}{4}$ and by Lemma \ref{lem-feasible} the solution $z_i\in P(\N_i)$. We now show the following claim.

\begin{claim}\label{cl:size}
Given a matroid $\N=(E,\I)$ and a vector $z\in P(\N)$ and a set $G\in \I$ such that for all $e\in G$, there exists a set $S$ containing $e$ such that $z(S) \geq r(S)-\frac12$. Then $\sum_{e\in E} z_e \geq \frac{|G|}{2}$.
\end{claim}
\begin{proof}
Consider the elements of $G$ in some order. When considering $e$, increase $z_e$ as much as possible while ensuring that $z$ remains in $P(\N)$. Call the new solution $y$. Observe that for each $e\in G$, there exists a set $S$ such that $y(S)=r(S)$. Consider the set of tight sets $\tau$ with respect to $y$ and let $S$ be the maximal set in $\tau$ as given by Fact~\ref{fact:uncross}. Then $\sum_{e\in E} y_e \geq  y(S)= r(S)\geq |G|$ since $G\in \I$ and $G\subseteq S$. But $\sum_{e\in E} z_e\geq \sum_{e\in E} y_e -\sum_{e\in G} \frac12$ since only variables in $G$ are increased to a fraction of at most half. Thus we have $\sum_{e\in E} z_e \geq \frac{|G|}{2}$ as required.
\end{proof}

Applying the claim for each $G_i$ since $G_i \in \I(\N_i)$, we obtain that for each $i$,
\begin{equation*}
\sum_{e \in E} z_{i,e}\geq |G_i|/2 \implies \sum_{i=1}^{n}\sum_{e \in E}z_{i,e} \geq \frac{n_\alpha \cdot \alpha}{16}
\end{equation*}

Since Lemma \ref{lem-feasible} implies that for each element $e$, $\sum_{i=1}^{n} z_{i,e} \leq \alpha x_e$, we have

$\alpha \sum_{e\in E} x_e \geq  \sum_{e\in E} \sum_{i=1}^{n} z_{i,e} \geq \frac{n_\alpha \cdot \alpha}{16}$, which concludes the proof.
\end{proof}

Finally, we prove a lemma bounding the performance of the algorithm.

\begin{lemma}\label{lem:guess}
For any guess value $\alpha$, the algorithm maintains a fractional solution to $LP_2(\alpha)$ such that:
$$ \sum_{e\in E}\sum_{i=1}^{n}z_{i,e}(\alpha) = \Omega\left(\frac{OPT_\alpha}{\log m}\right),$$ \label{main-eqn}
where $OPT_\alpha$ is objective of an optimal integral solution to $LP_2(\alpha)$.
\end{lemma}
\begin{proof}
Let $(x^*,z^*)$ denote the optimal integral solution to $LP_2(\alpha)$. If $x^*_e=0$ for each $e$, then the lemma follows immediately.
We have the following
\begin{displaymath}
\begin{array}{cccl}
\sum_{e\in E}\sum_{i=1}^{n}z_{i,e}(\alpha) & \geq & \frac{\alpha}{48\log m}\sum_{e\in E}\left(x_e(\alpha)-\frac{1}{m^2}\right) & ( Lemma ~\ref{obsxz})\vspace{0.2cm}\\
& \geq & \frac{\alpha}{48\log m}\sum_{e \in E} \left(\frac{x^*_e(\alpha)}{16} - \frac{1}{m^2}\right) &  (Lemma~\ref{lem:fractional-value})\vspace{0.2cm}\\
& = & \Omega\left(\frac{1}{\log m}\sum_{e\in E}\sum_{i=1}^{n}z_{i,e}^*(\alpha)\right) &\label{ineqf3}\vspace{0.2cm}
\end{array}
\end{displaymath}

where the last equality follows since in $LP_2(\alpha)$ for each element $\sum_{i=1}^{n}z_{i,e}^*(\alpha) \leq \alpha x^*_e$ and $\sum_{e\in E} x^*_e\geq 1$. This completes the proof of Lemma.
\end{proof}

Finally, we get our main theorem.
\begin{theorem}
The online algorithm for the fractional LP solution (of $LP_1$) is $O(\log m \log n)$-competitive.
\end{theorem}

\begin{proof}
The proof follows by combining Lemma (\ref{lem-feasible}), Lemma (\ref{lem:guess}), Lemma (\ref{lem:reduction}) and the observation that there are $O(\log n)$ possible values of $\alpha$ each is guessed with probability $\Omega(1/\log n)$.
\end{proof}

\section{Randomized Rounding Algorithm}\label{sec:randomized}

In this section we present a randomized algorithm for the unweighted problem that is $O(\log ^2 n \log m)$-competitive when each submodular function $f_i$ is a rank function of a matroid.
The algorithm is based on the fractional solution designed in Section \ref{sec:frac}. Although our rounding scheme is extremely simple, the proof of its correctness involves carefully matching the performance of the rounding algorithm with the performance of the fractional algorithm. Indeed, here the fact that $LP_2(\alpha)$ has extra constraints not present in $LP_1$ is used very crucially.

\begin{theorem}\label{theorem:rounding}
The expected profit of the randomized algorithm is $\Omega\left(\frac{OPT}{\log m \log^2 n}\right)$.
\end{theorem}

The randomized algorithm follows the following simple rounding procedure.

\vspace{0.15cm}
\noindent \framebox[1.05\width][l]{
\begin{minipage}{0.94\linewidth}
{\bf Matroid Randomized Rounding Algorithm:}
\begin{itemize}\setlength{\itemsep}{1pt}\setlength{\parskip}{0pt}
\item $F\gets \emptyset$.
\item Guess the value $\alpha\in_R\{1, 2, 4 \ldots, n\}$.
\item Run AlgG with value $\alpha$.
\begin{itemize}
\item Whenever $x_e$ increases by $\Delta x_e$, if $F\cup \{e\}\in \I(\M)$ then $F\gets F\cup \{e\}$ with probability $\frac {\Delta x_e}{4}$.
\end{itemize}
\end{itemize}
\end{minipage}}
\vspace{0.15cm}

In order to prove our main theorem, we prove several crucial lemmas. The main idea is to tie the performance of the randomized algorithm to the performance of the fractional solution that is generated. In the process we lose a factor of $O(\log n)$.
We first introduce some notation. All of the following notation is with respect to the execution of the online algorithm for a fixed value of $\alpha$ and we omit it from the notation. Let $F_i$ denote the solution formed by the randomized algorithm at the end of iteration $i$ and let $F$ denote the final solution returned by the algorithm. Let $Y^i_e$ denote the indicator random variable that element $e$ has been selected \emph{till} iteration $i$. Let $\Delta Y^i_e$ denote the indicator random variable that element $e$ is selected in iteration $i$. Let $y^i_e=Pr[Y_e^i=1]$ and $\Delta y^i_e=Pr[\Delta Y_e^i=1]$. Finally, let $y_e$ denote the probability element $e$ is in the solution at the end of the execution.
Recall that $x_{i,e}$ denotes the value of the variable $x_e$ in the fractional solution after iteration $i$ and let $x_e$ denote the fractional value of element $e$ at the end of the execution of the fractional algorithm, and let $\Delta x_{i,e}$ be the change in the value of $e$ in iteration $i$.

Since the algorithm tosses a coin for element $e$ in iteration $i$ with probability $\Delta x_{i,e}/4$, therefore the probability that an element $e$ is included in the solution till iteration $i$ is at most $x_{i,e}/4$.
Our first lemma states that the expected number of elements chosen by the algorithm is at least half that amount in expectation and is comparable to the total size of the fractional solution. Thus, Lemma~\ref{lem:integral-size} plays the role of Lemma~\ref{lem:fractional-value} in the analysis of the randomized algorithm. 

\begin{lemma}\label{lem:integral-size}
Let $F$ be the solution returned by the randomized rounding algorithm, then $E[|F|] = \sum_{e \in E}y_e \geq \frac{\sum_{e\in E}x_e}{8}$.
\end{lemma}
\begin{proof}
For every element $e$ in the $i^{th}$ iteration, the algorithm tosses a coin with probability $\Delta x_{i,e}/4>0$ and includes it in $F$ if the element shows up in the toss and $F\cup \{e\}$ is in $\I(\M)$ for the current $F$. We maintain a set $H$ in which we include an element which shows up in the toss but cannot be included in $F$. For any iteration $i$ and element $e$, we have $Pr[e\in F] +Pr[e\in H]=\Delta x_{i,e}/4$. Let $F$ denote the final set thus obtained. We now do the following random experiment. For every element $e\in span(F)\setminus F$, we toss a coin with probability $\frac{x_e}{4}$ and include it in a set $H'$. Since every element in $H$ must be in $span(F)$, by a simple coupling of the random tosses in these two experiments, we obtain $|H|\leq |H'|$. Moreover, $E[|H|]\leq E[|H'|]=\sum_{e\in span(F)}x_e/4\leq |F|$ where the expectation is taken over the random tosses used for constructing $H'$ since $x \in P_{\M}$ and $F\in \I(\M)$ . Taking expectations over the coin tosses used for finding $F$, we obtain $E[|H|]\leq E[|F|]$. But $E[|F|]+E[|H|]=\sum_{e\in E} x_e/4$. Thus $E[|F|]\geq  \frac{\sum_{e\in E}x_e}{8}$.
\end{proof}

Our second lemma relates the change in the probability an element is chosen to the change in the fractional solution. This lemma shows that a crucial property of the exponential update rule for the fractional solution is also satisfied by the integral solution. 
\begin{lemma}\label{deltay}
For each element $e$ and iteration $i$,
$\frac{\Delta y^i_e}{y^i_e} \leq \frac{\Delta x_{i,e}}{x_{i,e}} \leq \frac{24\log m}{\alpha}$.
\end{lemma}
\begin{proof}
First, we prove that $\frac{\Delta x_{i,e}}{x_{i,e}} \leq \frac{8\log m}{\alpha}$. This follows from the exponential update rule.
If $\alpha \leq 8\log m$, then we have $\frac{\Delta x_{i,e}}{x_{i,e}} \leq 1 \leq \frac{24\log m}{\alpha}$.
Otherwise,
$$\frac{\Delta x_{i,e}}{x_{i,e}}  \leq \frac{x_e \cdot \exp \left(\left(\frac{8 \log m}{\alpha}\right)-1\right)}{x_e} \leq \frac{24 \log m}{\alpha}  $$
where the second inequality follows since $e^x-1 \leq 3x$ for $0 \leq x \leq 1$.

We next show that $\frac{\Delta y^i_e}{\Delta x_{i,e}} \leq \frac{ y^i_e}{x_{i,e}}$ implying the claim.
To show  $\frac{\Delta y^i_e}{\Delta x_{i,e}} \leq \frac{ y^i_e}{x_{i,e}}$, it suffices to show  $\frac{\Delta y^i_e}{\Delta x_{i,e}} \leq \frac{\Delta y^j_e}{\Delta x^j_e}$ whenever $j\leq i$ since $\frac{ y^i_e}{x_{i,e}} =  \frac{\sum_{j\leq i} \Delta y^j_e}{\sum_{j \leq i}\Delta x^j_e}$.

Now, observe that for any $j$, we have $\Delta y^j_e= Pr[e \textrm{ is included in }F \textrm{ in the $j^{th}$ iteration}]=\Delta x^j_e\cdot Pr[F_{j-1} \cup \{e\} \in \I(\M)]$. Therefore, $\frac{\Delta y^j_e}{\Delta x^j_e}=Pr[F_{j-1} \cup \{e\} \in \I(\M)]$ which is a decreasing function of $j$ since $F_{j-1}$ is an increasing set as a function of $j$.
\end{proof}

We next prove a general lemma regarding randomized rounding in any matroid polytope. The proof of the lemma utilizes a lemma proven in Chawla et al.~\cite{CHMS10}.

\begin{lemma}\label{lem:matroid-cover}
Given a matroid $\N=(E,\I)$ and a solution $z$ such that for all $S \subseteq E$, $z(S) \leq r(S)/2$, construct a set $F$ by including in $e\in F$ with probability $z_e$ for each $e\in E$ independently. Then, with high probability $(1-\frac{1}{m^2n^2})$, $F$ can be covered by $O(\log m+\log n)$ independent sets where $m=|N|$.
\end{lemma}
\begin{proof}
The lemma follows by a coupling argument on the following process. Define the following process:
\begin{itemize}
\item Set $S \gets E$, set $I_1, I_2, \ldots, I_m \gets \emptyset$.
\item As long as $S$ is not empty iterate $i=1, 2, \ldots, m$.
\begin{itemize}
\item In each iteration $i$, order the elements by any order (we later define a good ordering).
\item Go over the elements in $S$ one-by-one according to the ordering. For each element $e$, if $I_i \cup \{e\}$ is independent in $\N$:
    \begin{itemize}
\item Toss a coin for $e$ and include it in $I_i$ with probability $z_e$.
\item $S \gets S \setminus \{e\}$.
\end{itemize}
\item Otherwise, skip $e$ in iteration $i$.
\end{itemize}
\end{itemize}

Observe that we always remove the first element in the ordering in each iteration from $S$, so there are at most $m$ iterations.
Also, note that each $I_i$ is an independent set. We first claim that this process is equivalent to the process of tossing a coin for each element independently. That is, for each $S' \subseteq E$:

\begin{eqnarray}
Pr[I_1 \cup I_2 \cup \ldots \cup I_m = S'] = \prod_{e \in S'}z_e \prod_{e \notin S'}(1-z_e)
\end{eqnarray}

This is very simple. For each element $e$, there exists a single iteration $1 \leq j \leq m$, in which we toss a coin for $e$. In this iteration it is included with probability $z_e$ that is independent of all other choices.
Thus, the process actually defines a covering by independent sets of the process of selecting each element independently with probability $z_e$.
If we can prove that with high probability $S$ is empty after only a logarithmic number of rounds we are done. This might not be true for any ordering, however, we prove that there exists some ordering for which this is true. The following lemma essentially follows from Lemma 7 in \cite{CHMS10}. We include a proof here for completeness.

\begin{claim}\label{cl:order}
Given a matroid $\N=(E,\I)$ with a rank function $r$ and a fractional solution $z$ such that for all $S \subseteq E$, $x(S) \leq r(S)/2$, there exists an ordering of the elements, such that for all $e \in E$, the probability we toss a coin for $e$ is at least $1/2$.
\end{claim}

\begin{proof}
We prove that there always an element $e$ that we can put last and if $I$ is the independent set constructed by including rest of the elements in any order $I\cup \{e\}\in \I(\N)$ with probability at least $\frac12$.  The claim then follows from recursing on rest of the elements.

Suppose we put some $e \in E$ last in the order. Let $D$ be the set constructed by selecting each element with probability $z_e$. When we reach $e$, we will toss a coin for $e$ only if $e \notin span(D\setminus \{e\})$.
Let $p_e$ be the probability $e$ is being tossed when it is put last. Thus,

$$p_e = Pr[e \notin span(D\setminus \{e\})] \geq Pr[e \notin span(D)]$$

Let $B$ be any base for the matroid and let $k$ be the rank of the matroid.
Then we know that:

\begin{eqnarray}
\sum_{e \in B} p_e & \geq &  \sum_{e\in B} Pr[e \notin span(D)] \label{in-set1} \\
& = & \sum_{D}Pr[D]\left(k-rank(D)\right) \nonumber \\
& = & k- E[\mbox{rank}(D)] \nonumber \\
& \geq & k-E[|D|] \geq k/2 \label{in-set2}
\end{eqnarray}
Inequality (\ref{in-set2}) follows since by our assumption $E[|D|]=\sum_{e\in E} z_e \leq k/2$.
Averaging we obtain that there exists an element in $B$ such that $p_e\geq \frac12$. We put this element last in the order and recurse.
\end{proof}

Suppose in each iteration we choose an ordering given by Claim~\ref{cl:order}. Then for each element $e \in E$ and round $i$, the probability it ``survives" in $S$ given that it ``survived" previous iterations is at most $1/2$. Thus,
the probability there exists an element in $S$ after $O(\log m+\log n)$ iterations is at most $1/(mn)^c$ for some large constant $c$.
\end{proof}

We now prove a relation between the profit obtained by the algorithm at iteration $i$, denoted by the random variable $r_i(F_i)$, and the events that a particular set of elements are chosen in the solution. For any $i$, Let $H_i$ denote the set of elements such that $z_{i,e}>0$. Note that $z_{i,e}>0$ iff $\Delta x_{i,e} >0$.

\begin{lemma}\label{lem:integral-copies}
$\sum_{i=1}^n E[r_i(F_i)] \geq \frac{1}{c\log n} \sum_{i=1}^n \sum_{e\in H_i} y^i_e$, where $c$ is some constant.
\end{lemma}

\begin{proof}
Let $F_i'=F_i\cap H_i=\{e\in F_i: z_{i,e} >0\}$. In other words, $F_i'$ is the set of elements included by the algorithm in iteration $i$. Observe that each element $e$ is included in $F_i'$ with probability at most $x_{i,e}/4 = z_{i,e}/2$ and is dropped sometime even when the toss comes up correct if adding $e$ violates the independence constraints for matroid $\M$. Now consider the set $F_i''$ where each element is included with probability $z_{i,e}/2$ with no regard to the matroid constraints. Coupling the tosses of the two random experiments, we get  $F'_i\subseteq F''_i$.
Since $z_i\in P(\N_i)$ by Lemma \ref{lem-feasible}, Lemma~\ref{lem:matroid-cover} implies that $F_i''$ is covered by at most $O(\log (mn))$ independent sets with high probability $1-1/(mn)^c$ for some large constant $c$.
Let $A_i$ be the event $F_i'$ can be covered by at most $O(\log (mn))$ independent sets in $\N_i$, where $m= |E|$. Therefore, we get that:

$$E[r_i(F'_i) | A_i] \geq \frac{E[\sum_{e\in F_i'}Y_e^i |A_i]}{\log (mn)}$$

Let $A$ be the event that for all $i$, $F'_i$ can be covered by $O(\log(mn))$ independent sets.
Then,
\begin{small}
\begin{eqnarray*}
 \sum_{i=1}^n E[r_i(F_i)] & \geq & \sum_{i=1}^n  E[r_i(F'_i)] \geq  \sum_{i=1}^n  E[r_i(F'_i) | A] \cdot \Pr[A]
\geq \sum_{i=1}^n \frac{E[\sum_{e\in H_i}Y_e^i |A]}{c \log (m n)} \cdot \Pr[A] \\
 &\geq & \frac{1}{c\log (m n)}\left(E[\sum_{i=1}^n \sum_{e\in H_i} Y^i_e] - nm \cdot Pr[\bar{A}]\right)
 = \Omega\left(\frac{1}{\log (m n)}\sum_{i=1}^n \sum_{e\in H_i}y^i_e\right)
\end{eqnarray*}
\end{small}
Here the last inequality follows since $1 \leq \sum_{i=1}^n \sum_{e\in H_i}y^i_e \leq m \cdot n$ and $\Pr[\bar{A}] \leq \frac{1}{{(mn)}^{c'}}$ for some constant $c'$. (If $\sum_{i=1}^n \sum_{e\in H_i}y^i_e<1$ then the statement follows easily).
\end{proof}
Now we have all the ingredients to prove Theorem~\ref{theorem:rounding}.

\begin{proofof}{Theorem~\ref{theorem:rounding}}
We prove that the expected profit of the algorithm with a guess $\alpha$ is at least $\Omega\left(\frac{OPT_\alpha}{\log m \log n}\right)$. Since each $\alpha$ is guessed with probability $1/\log n$ and the value of OPT is the sum over all values $\alpha$ we get the desired. The expected profit of the algorithm when we guess $\alpha$ is at least.
\begin{displaymath}
\begin{array}{ccll}
  \sum_{i=1}^n E[ f_i(F_i)] &\geq &  \frac{1}{c\log n}\sum_{i=1}^n  \sum_{e\in H_i} y^i_e \qquad &(Lemma~\ref{lem:integral-copies}) \vspace{0.2cm} \\
 &\geq & \sum_{i=1}^n \sum_{e\in H_i} \frac{\alpha}{c'\log m\log n} \Delta y^i_e & (Lemma~\ref{deltay}) \vspace{0.2cm}\\
 &=&\sum_{e\in E}\frac{\alpha}{c'\log m \log n}  y_e   & (\sum_{i: e\in H_i}\Delta y_e^i=y_e) \vspace{0.2cm}\\
  &\geq & \sum_{e\in E} \frac{\alpha}{8c'\log m \log n} x_e  \qquad \qquad &(Lemma~\ref{lem:integral-size}) \vspace{0.2cm} \\
   & = & \Omega\left(\frac{\alpha \cdot n_\alpha}{\log m \log n}\right) = \Omega\left(\frac{OPT_\alpha}{\log m \log n}\right) &(Lemma~\ref{lem:fractional-value})
\end{array}
\end{displaymath}
\end{proofof}

\section{Extensions}\label{sec:extensions}

In this section, we give simple reductions to deal with various extensions including the case of arbitrary weighted rank functions and the case when the algorithm does not the number of arrivals $n$ in advance. 

\subsection*{Extension to the weighted version}\label{sec:weighted}
In this section we show a simple reduction from the weighted rank function case to the unweighted case losing an additional factor of $O(\log f_{ratio})$. The same reduction should be used when considering the randomized algorithm in Section \ref{sec:randomized}.
We first deal with the case in which $f_{ratio}$ is known in advance. In this case, to solve the weighted version, first guess a value $j\in \{0,1, \ldots, f_{ratio}\}$. Let $f_{min}=\min_{i,e:f_{i}(\{e\})>0} f_i(\{e\})$. Then for each $i$, replace $f_{i}$ with the following $f_{ij}$ which has a restriction on the elements' weights.
$f_{ij}(S)=\max_{I\subseteq S, I\in \I(\M)}|\{e\in I, 2^j\cdot f_{min}\leq w_{i,e} < 2^{j+1}\cdot f_{min}\}|$.
It is clear that for any $f_i$ and any $S$:
$$f_i(S) \leq 2 \cdot \sum_{j=0, 1, \ldots, f_{ratio}}2^j \cdot f_{min}\cdot f_{ij}(S)$$
As for each $j$ we get a competitive algorithm, we lose an extra $O(\log f_{ratio})$ due to the guess at the beginning.

If $f_{ratio}$ is unknown, we modify our guess slightly. We guess value $2^i$ with probability $\frac{1}{c i (\log(1+i))^{1+\epsilon}}$, for a suitable $c$ that depends on $\epsilon$.
We then make the same changes as before with this value. It is easy to verify that with a suitable $c$, the sum of probabilities over all values $i$ is at most $1$. Moreover, the probability of each weight $w$ class is $\Omega(\frac{1}{\log w (\log \log w)^{1+\epsilon}} = \frac{1}{\log f_{ratio} (\log \log f_{ratio})^{1+\epsilon}}$, and so we lose an extra $O(\log f_{ratio} (\log \log f_{ratio})^{1+\epsilon})$ due to the guess.

\subsection*{Dealing with an unknown $n$}\label{sec:double}
In many cases the number of submodular functions $n$ is unknown in advance. We show here a simple modification to the algorithm that can handle such cases while losing an additional $\log \log n$ factor.
The idea is simply to replace the guessing stage in the algorithm. In the modified guessing stage, we guess value $\alpha= 2^i$ with probability $\frac{1}{c i (\log(1+ i))^{1+\epsilon}}$, where $c$ is a constant that depends on $\epsilon$. We then run the algorithm with our guess.
First, note that (with a suitable value $c$) the total sum of probabilities of all values $\alpha$ is bounded.

$$\sum_{i=1}^{\infty}\frac{1}{c \cdot i (\log(1+ i))^{1+\epsilon}} <1$$

To analyze the performance, assume that the number of submodular functions is $n$, then by Lemma \ref{lem:guess}, for any $\alpha\leq n$ our profit is $\Omega\left(\frac{OPT_\alpha}{\log m}\right)$. Thus, the total expected profit of the algorithm is at least:

$$\sum_{i=1}^{\log n} \frac{1}{c \cdot i (\log(1+ i))^{1+\epsilon}}\frac{OPT_\alpha}{\log m} \geq \sum_{i=1}^{\log n} \frac{1}{c \cdot \log n (\log \log n)^{1+\epsilon}}\frac{OPT_\alpha}{\log m} = \Omega\left(\frac{OPT}{\log m \log n (\log \log n)^{1+\epsilon}}\right)$$

The same strategy also applies to the randomized rounding algorithm.

%

\vspace{-0.3cm}
\bibliographystyle{plain}
\bibliography{matroidbib}

\end{document}